%% file: main.tex
\documentclass[11pt]{article}
\input{header.tex}

\usepackage[numbers]{natbib}

\usepackage[USenglish]{babel}
\begin{document}

\title{$\lambda$-Regularized $A$-Optimal Design and its Approximation by $\lambda$-Regularized Proportional Volume Sampling}
\author{Uthaipon (Tao) Tantipongpipat\thanks{Twitter. The work was done while the author was at Georgia Institute of Technology.}}
\date{\today}
\maketitle

\begin{abstract}
    In this work, we study the \textit{$\lambda$-regularized $A$-optimal design} problem  and introduce the \textit{$\lambda$-regularized proportional volume sampling} algorithm, generalized from  [Nikolov,  Singh, and Tantipongpipat, 2019], for this problem with the approximation guarantee that extends upon the previous work. In this problem, we are given vectors $v_1,\ldots,v_n\in\mathbb{R}^d$ in $d$ dimensions, a budget $k\leq n$,  and the regularizer parameter $\lambda\geq0$, and the  goal is to 
find a subset $S\subseteq [n]$ of size $k$ that minimizes the trace of $\left(\sum_{i\in S}v_iv_i^\top + \lambda I_d\right)^{-1}$
where $I_d$ is the $d\times d$ identity matrix.
The problem is motivated from optimal design in  ridge regression, where one tries to minimize the expected squared error of the ridge regression predictor from the true coefficient in the underlying linear model.
We introduce  $\lambda$-regularized proportional volume sampling and give its polynomial-time implementation to solve this problem. We show its  $(1+\frac{\epsilon}{\sqrt{1+\lambda'}})$-approximation for $k=\Omega\left(\frac d\epsilon+\frac{\log 1/\epsilon}{\epsilon^2}\right)$ where $\lambda'$ is proportional to $\lambda$, extending the previous bound in [Nikolov,  Singh, and Tantipongpipat, 2019] to the case $\lambda>0$ and obtaining asymptotic optimality as $\lambda\rightarrow \infty$. 

\end{abstract}


\section{Introduction}
Optimal design is a classical problem in statistics \cite{fedorov1972theory} with many applications from diversity sampling to machine learning. Optimal design has many different criteria, such as \textit{A,D,E,V}-optimality,   which correspond to different objectives to be optimized.  In this work, we focus in \(A\)-optimality. We refer the reader to \cite{nikolov2019proportional} and references therein for applications of optimal design and other optimality criteria. 

The problem of \(A\)-optimal design can be defined as follows. We are given \(n\) input vectors \(V=\set{v_1,\ldots,v_n}\) where \(v_i\in\R^d\) is in \(d\) dimensions and a budget \(k\leq n\),  and the goal is to find a subset \(S\sbq [n]\) of size \(k\) that minimizes the trace of $\pr{\sum_{i\in S}v_iv_i^\top }^{-1}$ (if \(V\) does not span full rank, we ignore the \(d-\rank(V)\) zero eigenvalues in calculating harmonic mean of the eigenvalues of
\(\sum_{i\in S}v_iv_i^\top\)). Approximation algorithms for \(A\)-optimal design include \(\frac{n-d+1}{k-d+1}\)-approximation by volume sampling \cite{avron2013faster}, $(1+\epsilon)$-approximation for $k=\Omega( \frac{d^2}{\epsilon})$ by a connection of optimal design with matrix sparsification \cite{wang2017computationally}, \((1+\eps)\)-approximation for \(k=\Omega(\frac d{\eps^2})\) by regret minimization \cite{allen2020near}, and \((1+\eps)\)-approximation for \(k= \Omega\left(\frac{d}{\epsilon^4}\right)\) and for \(k=\Omega\pr{\frac{d}{\eps^3} \log^2 \frac 1\eps}\) using a variant of local search and greedy algorithms \cite{madan2019combinatorial}. The best approximation known in the  regime with large \(k\)   is obtained by \cite{nikolov2019proportional} as follows.
\begin{theorem}[\cite{nikolov2019proportional}]
There exists a polynomial-time \((1+\eps)\)-approximation algorithm for \(A\)-optimal design problem for \(k=\Omega\pr{\frac d\eps+\frac{\log 1/\eps}{\eps^2}}\).
\end{theorem}
The result follows from solving the convex relaxation of \(A\)-optimal design and sampling a set with \textit{proportional volume sampling} based on the fractional solution obtained from the relaxation. \citet{nikolov2019proportional} show that approximation guarantee of \(A\)-optimal design follows from \(\emph{approximately independent}\) distribution and that a general class of \textit{hard-core} distributions is  approximately independent. Finally, they show that \textit{proportional volume sampling} can be efficiently implemented and is, indeed, a hard-core distribution, which conclude the proof of  the approximation.
 
In this work, we generalize this approach to the  \textit{\(\lambda\)-regularized \(A\)-optimal design} problem, where one aims to minimizes the trace of $\pr{\sum_{i\in S}v_iv_i^\top +\lambda I_d}^{-1}$ where \(I_d\) is the \(d\times d\) identity matrix. The problem is motivated from the use of ridge regression, a variant of linear regression with an \(\l_2\)-regularization penalty, to find the best linear estimator. We define \textit{near-pairwise independent} distributions, and show that they also include a general class of hard-core distributions, and that near-pairwise independence implies approximation guarantee for \(\lambda\)-regularized \(A\)-optimal design. Finally, we define \textit{\(\lambda\)-regularized proportional volume sampling} and show its near-pairwise independence property and its polynomial-time implementation. All of these results imply the approximation to \(\lambda\)-regularized \(A\)-optimal design, which is our main result and is stated as follows.

\begin{theorem}
There exists a polynomial-time \((1+\eps)\)-approximation algorithm for \(\lambda\)-regularized \(A\)-optimal design problem for \(k=\Omega\pr{\frac d\eps+\frac{\log 1/\eps}{\eps^2}}\).
In fact, the approximation ratio is \((1+\eps(\lambda))\) where \(\eps(0)=\eps\) and \(\eps(\lambda)\rightarrow0\) as \(\lambda\rightarrow\infty\). 
\end{theorem}
The exact approximation ratio and constants in the bound of \(k\) can be found in Theorem \ref{thm:reg-ratio-with-lambda}. Our analysis follows similarly as the one in \cite{nikolov2019proportional}, which heavily involves elementary symmetric polynomials of eigenvalues of the matrix \(\sum_{i\in S}v_iv_i^\top\). The key idea in extending the previous results to \(\lambda\)-regularized \(A\)-optimal design is the fact that an elementary symmetric polynomial of eigenvalues of \(\sum_{i\in S}v_iv_i^\top+\lambda I_d\) are sums of  elementary symmetric polynomials of eigenvalues of \(\sum_{i\in S}v_iv_i^\top\). We then carefully group these polynomials and  bound each of those groups using similar but more complicated inequalities from \cite{nikolov2019proportional}.

\subsection{Related Work}
For related work to \(A\)-optimal design and its approximation algorithms, we refer the reader to \cite{nikolov2019proportional} and references therein. Here, we focus on work related to \(\lambda\)-regularized \(A\)-optimal design, when one uses ridge regression in place of linear regression to find a linear estimator in optimal design.

Ridge regression or regularized regression is introduced by \citet{hoerl1970ridge} to ensure a unique solution of linear regression when a data matrix is singular, i.e., when  the training data points do not span full \(d\) dimensions. Ridge regression has  been applied to many practical problems \cite{marquardt1975ridge} and is one of  classical linear methods for regression in machine learning \cite{hastie2009elements}.

\citet{derezinski2017subsampling} introduced \textit{\(\lambda\)-regularized  volume sampling}, and their results imply \(\frac {n}{k-d+1}\)-approximation for \(\lambda\)-regularized \(A\)-optimal design. The linear dependence on \(n\) in the approximation ratio is a result of their bound of $\tr\pr{\sum_{i\in S}v_iv_i^\top +\lambda I_d}^{-1}$ that compares to \(\tr\pr{\sum_{i\in[n]}v_iv_i^\top +\lambda I_d}^{-1}\) rather than to \( \tr\pr{\sum_{i\in S^*}v_iv_i^\top +\lambda I_d}^{-1}\) for an optimal \(S^*\subseteq [n]\) of the problem as in our work. We compare their result to ours in more details in Appendix \ref{sec:compare}. 

\subsection{Organization}
In Section \ref{sec:background}, we provide background on optimal design and the motivation and definition of the  {\(\lambda\)-regularized \(A\)-optimal design} problem. In Section \ref{sec:alg}, we describe our algorithm based on convex relaxation and \(\lambda\)-regularized proportional volume sampling. In Section \ref{sec:reduct}, we state near-pairwise independence property and prove its sufficiency to approximate \(\lambda\)-regularized \(A\)-optimal design. In Section \ref{sec:reg-prop-of-mu}, we show that \(\lambda\)-regularized proportional volume sampling is hard-core, and that hard-core distributions are near-pairwise independent. In Section \ref{sec:reg-combine-ratio}, we state and prove our main technical result, namely the approximation of \(\lambda\)-regularized \(A\)-optimal design. In Section 7, we show a polynomial-time implementation of \(\lambda\)-regularized proportional volume sampling. We note in Appendix \ref{sec:compare} the comparison of  \(\lambda\)-regularized  volume sampling \cite{derezinski2017subsampling,derezinski2018reverse} and our \(\lambda\)-regularized proportional volume sampling. Appendix \ref{sec:cal-error}  contains  derivations of formula deferred from the main body.

\section{Notation, Background, and Motivation of \(\lambda\)-Regularized \textit{A}-Optimal Design} \label{sec:background}
Let \(V=[v_1 \ldots v_n]\) be the \(d\)-by-\(n\) matrix of input vectors \(v_i\in \R^d\). We use the notation $x^S = \prod_{i \in S}x_i$, $V_S$ a matrix of column vectors $v_i\in \R^d$ for $i\in S$, and $V_S(x)$ a matrix of column vectors $\sqrt{x_i} v_i\in \R^d$ for $i\in S$. Let \(y\) be the label (or response) column vector, and \(y_S\) is the \(k\times 1\) column vector \((y_i)_{i \in S}\). Denote \(\U_k,\U_{\leq k}\) the sets of all subsets of \([n]\) of size \(k\) and at most \(k\), respectively. Let $e_k(x_1, \ldots, x_n)$ be the degree $k$
elementary symmetric polynomial in the variables $x_1, \ldots,
x_n$, i.e.,
$e_k(x_1, \ldots, x_n) = \sum_{S \in \U_k} x^S$.
By convention, $e_0(x) = 1$ for any $x$ {and $e_k(x_1, \ldots, x_n)=0$ for $k>n$}. For any positive
semi-definite $n\times n$ matrix $M$, we define $E_k(M)$ to be
$e_k(\lambda_1, \ldots, \lambda_n)$, where $\lambda(M)=(\lambda_1, \ldots,
\lambda_n)$ is the vector of eigenvalues of $M$. Denote \(I_n\) the identity matrix of dimension \(n\times n\), \(Z_S(\lambda)=V_SV_S^\top+\lambda I_d\), and \(\an{A,B}\) the dot product of two matrices \(A,B\) of the same dimension.
We denote \(\NN(\mu,\Sigma)\) the multi-variate Gaussian distribution with mean \(\mu\) and covariance \(\Sigma\).

Different optimality criteria of optimal design can be viewed as different scalarizations of the matrix  \(V_SV_S^\top\), such as the trace of the inverse as in \(A\)-design, or the determinant in \(D\)-design. One motivation on which we focus in this work for \(A\)-design is  the squared error of the estimator in linear model. In linear model, we assume that \(y_i = v_i^\top w^* + \eta_i\) where  \(\eta_i\)'s are independent Gaussian noise with mean zero and  variance \(\sigma^2\). 
We want to pick \(S\subseteq[n]\) to obtain labels \(y_S\) which provide as much information as possible to best estimate \(w^*\). 

\paragraph{Linear Regression.} 
One choice to estimate $w^*$ is by minimizing the sum of squared errors on the labeled samples:
\begin{equation}
\hat{w}_S= \argmin_{w\in\R^d}\bc{\norm{y_S-V_S^\top w}{2}^2}
\end{equation}
which is also called \textit{linear regression}. This estimate is also known to be the maximum likelihood estimate (with no prior). The expected squared error \(\Ex{\eta_S}{\norm{\hat{w}_S - w^*}{2}^2}\) of this estimator \(\hat{w}_S\) from \(w^*\) is \(\sigma^2\tr{(V_SV_S^\top)^{-1}}\) (see Appendix \ref{sec:cal-error} for its derivation). Hence, to get as useful predictor \(\hat w_S\) as possible, one can  minimize \(\tr{(V_SV_S^\top)^{-1}}\), which is a motivation to the \(A\)-design objective.

\paragraph{Ridge Regression.}
Suppose we estimate $w^*$ by minimizing the sum of squared errors on the labeled samples with an additional \(\ell_2\)-regularization parameter \(\lambda\):
\begin{equation}
\hat{w}_S(\lambda)= \argmin_{w\in\R^d}\bc{\norm{y_S-V_S^\top w}{2}^2+\lambda \norm{w}{2}^2}
\end{equation}
which is also called \textit{ridge regression}. Ridge regression with \(\lambda>0\) increases the stability the linear regression against the outlier, and forces the optimization problem to have a unique solution  when \(V\) does not span full-rank \(d\) which makes linear regression  ill-defined.
When \(\lambda=0\), the problem reverts to standard linear regression.
It is also known that 
\(\hat{w}_S(\lambda)\)
is the maximum likelihood estimate of linear model given the Gaussian prior \(w^*\sim \mathcal{N}(0,\frac{\sigma^2}{\lambda} \cdot I_d)\). 
The expected squared error of \(\hat{w}_S (\lambda)\) from  \(w^*\) is 

\begin{equation}
\Ex{\eta_S}{\norm{\hat{w}_S (\lambda)- w^*}{2}^2}= \sigma^2 \tr{ Z_S(\lambda) ^{-1}}
-\lambda\an{ Z_S(\lambda) ^{-2}, \sigma^2I_d-\lambda w^* {w^*}^\top}. \label{eq:ridge-error}
\end{equation} 

\begin{table}
\caption{Distributions of model (or predictor) and prediction errors of the ridge regression estimator \(\hat{w}_S(\lambda)\)}
\begin{center}
\begin{tabular}{|x{2cm}|x{5cm}|x{6cm}|}\hline
 Settings& \(\hat{w}_S(\lambda) - w^*\) & \(X^\top \pr{\hat{w}_S(\lambda) - w^*}\) \\\hline
\(\lambda=0\) & \(=\NN\pr{0,\sigma^2 \pr{V_S V_S^\top}^{-1}}\) & \(=\NN\pr{0,\sigma^2 X^\top \pr{V_SV_S^\top}^{-1} X}\) \\\hline
\(\lambda\geq 0\) &  \begin{tabular}{c}
\(=\NN(-\lambda Z_S(\lambda) ^{-1} w^*,\) \\
\( \sigma^2 \br{ Z_S(\lambda) ^{-1} - \lambda Z_S(\lambda) ^{-2} }) \) \
\end{tabular}
 &  \begin{tabular}{c}
\(=\NN(-\lambda X^\top  Z_S(\lambda) ^{-1} w^*,\) \\
\( \sigma^2 X^\top \br{ Z_S(\lambda) ^{-1} - \lambda Z_S(\lambda) ^{-2} } X) \) 
\end{tabular} \\\hline
\end{tabular}
\end{center}
\label{tab:gau-noise-reg-pred-err}
\end{table}
\begin{table}
\begin{center}
\caption{Expected squared error of model (or predictor)\ and prediction errors of the ridge regression estimator \(\hat{w}_S(\lambda)\)}
\label{tab:sqr-loss-reg-pred-err}
\begin{tabular}{|x{1.4cm}|x{5.7cm}|x{8.0cm}|}\hline
 Settings  & \(\Ex{\eta_S}{\norm{\hat{w}_S(\lambda) - w^*}{2}^2}\) & \(\Ex{\eta_S}{\norm{X^\top \pr{\hat{w}_S(\lambda) - w^*}}{2}^2} \) \\\hline
\(\lambda=0\) & \(=\sigma^2\tr{V_SV_S^\top}^{-1}\) & \(=\sigma^2\tr{ X^\top \pr{V_SV_S^\top}^{-1} X}\) \\\hline
\(\lambda\geq 0\) &  \begin{tabular}{c}
\(= \sigma^2 \tr{ Z_S(\lambda) ^{-1}}\) \\
\(-\lambda\an{ Z_S(\lambda) ^{-2}, \sigma^2I_d-\lambda w^* {w^*}^\top} \)
\end{tabular}
 &  \begin{tabular}{c}
\(=\sigma^2 \tr{ X^\top Z_S(\lambda) ^{-1}X}\) \\
\(  -\lambda\an{ Z_S(\lambda) ^{-1}X X^\top Z_S(\lambda)^{-1}   ,\sigma^2I_d-\lambda w^* {w^*}^\top}\)
\end{tabular} \\\hline
\end{tabular}
\end{center}

\end{table}

We summarize the distribution of the predictor or model error, \(\hat{w}_S (\lambda)- w^*,\) and the prediction error with respect to a data matrix \(X\) in \(d\) dimensions, \(X^\top \pr{\hat{w}_S(\lambda) - w^*}\), of the ridge regression estimate \(\hat{w}_S (\lambda)\) in Tables \ref{tab:gau-noise-reg-pred-err} and \ref{tab:sqr-loss-reg-pred-err}.
Some optimality criteria concern prediction error; for example, \(V\)-optimal design minimizes the expected squared norm of \(X^\top \pr{\hat{w}_S(\lambda) - w^*}\) with \(X=V\).
We note that in general, we may also assume \(\eta\) is a random Gaussian vector  \(\NN\pr{0,\Cov{\eta}}\)  with \(\Cov{\eta} \preceq \sigma^2 I_n\) (instead of \(\Cov{\eta} = \sigma^2 I_n\)), and the results in this work still hold; the errors to be minimized will be upper bounded by as if \(\eta\sim \NN\pr{0,\sigma^2 I_n}\). The derivation of Tables \ref{tab:gau-noise-reg-pred-err} and \ref{tab:sqr-loss-reg-pred-err} can be found in Appendix \ref{sec:cal-error}.

\paragraph{Bounding the Error of Ridge Regression Predictor.}

The challenge to upper-bound \eqref{eq:ridge-error} is the second-order term \(Z_S(\lambda) ^{-2}\). One way to address this is to consider only the first-order term \(\tr\pr{ Z_S(\lambda) ^{-1}}\). For example,  \citet{derezinski2017subsampling} assume that 
\( \lambda \leq \frac{\sigma^2}{\norm{w^*}{2}^2}\), which gives \(\lambda w^* {w^*}^\top \preceq \sigma^2 I\), and then we have
\begin{align}
\Ex{\eta_S}{\norm{\hat{w}_S(\lambda) - w^*}{2}^2}& \leq\sigma^2 \tr\pr{ Z_S(\lambda) ^{-1}}. \label{eq:model-error-final-bound}
\end{align} 
The right-hand side of \eqref{eq:model-error-final-bound} now contains only  the first-order term \(\tr\pr{ Z_S(\lambda) ^{-1}}\), which can be easier to optimize. For example,  results in \cite{derezinski2017subsampling,derezinski2018reverse}  imply an approximation for the objective \( \tr\pr{ Z_S(\lambda) ^{-1}}\).
To the best of our knowledge, it is an open question whether there is an approximation algorithm that directly bounds \(\Ex{\eta_S}{\norm{\hat{w}_S(\lambda) - w^*}{2}^2}\) without any assumption on \(\lambda\).

\subsection{\(\lambda\)-Regularized \(A\)-Optimal Design  }
The upper-bound \(\sigma^2 \tr\pr{ Z_S(\lambda) ^{-1}}\) of the expected squared predictor error in \eqref{eq:model-error-final-bound} is similar to the \textit{A}-optimal design objective \(\tr\pr{ V_SV_S^\top}^{-1}\) , and we follow \citet{derezinski2017subsampling} in using it as an objective to be optimized.
In particular, we define the \textit{\(\lambda\)-regularized \(A\)-optimal design} problem as,  given input vectors \(V=[v_1 \ldots v_n]\in \R^{d \times n}\) in \(d\) dimensions, positive integer \(k\), and \(\lambda \geq0\), we find a subset \(S\subseteq[n]\) of size \(k\) to minimize
\begin{align}
\min_{S\sbq [n], |S|=k} \tr{ \pr{V_SV_S^\top+\lambda I_d} ^{-1}}. \label{eq:obj-reg}
\end{align}
\paragraph{\(\lambda\)-regularized Generalized Ratio Objective.} Similar to the generalized ratio objective in \cite{nikolov2019proportional}, we can also  define its \(\lambda\)-regularized counterpart. The generalized ratio objective is the ratio of elementary symmetric polynomials of eigenvalues of \(V_SV_S^\top\), which captures both \(A\)- and \(D\)-design problems. Given \(0\leq l'\leq l\leq d\), the goal is to choose a subset \(S\subseteq[n]\) of size \(k\) to minimize
\begin{equation}
\min_{S\subseteq [n],|S|=k} \left(\frac{E_{l'}(V_S V_S^\top)}{E_l(V_S V_S^\top)}\right)^{\frac{1}{l-l'}} \label{eq:objective-gen}
.\end{equation}
Hence, one can also define \(\lambda\)-regularized generalized ratio objective as 

\begin{equation}
\min_{S\subseteq [n],|S|=k} \left(\frac{E_{l'}(V_S V_S^\top+\lambda I_d)}{E_l(V_S V_S^\top+\lambda I_d)}\right)^{\frac{1}{l-l'}}
.\end{equation}
\section{\(\lambda\)-Regularized Proportional Volume Sampling Algorithm} \label{sec:alg}

Recall that we denote \(\U_k\) (\(\U_{\leq k}\)) the set of all subsets \(S\sbq [n]\) of size \(k\) (of size \(\leq k\)). Given \(\lambda \geq 0, y\in \R^n,\U\in\set{\U_k,\U_{\leq k}}\), and \(\mu\) a distribution over \(\U\), we define the \textit{\(\lambda\)-regularized proportional volume sampling with measure \(\mu\)} to be the distribution \(\mu'\) over \(\U\) where \(\mu'(S)\propto \mu(S) \det Z_S(\lambda) \) for all \(S\in \U\). Given \(y\in \R^n\), we say a distribution \(\mu\) over \(\U\) is \textit{hard-core with parameter \(z\) }if \(\mu(S) \propto z^S:=\prod_{i \in S} z_i\) for all \(S\in \U\). Denote \(\norm{A}{2}\) the spectral norm of matrix \(A\). 

To solve \(\lambda\)-regularized \(A\)-optimal design, we solve the convex relaxation of the optimization problem, namely
\begin{align}
&\min_{x\in\R^n} \frac{E_{d-1} (V(x)V(x)^\top+\lambda I) }{E_{d} (V(x)V(x)^\top+\lambda I)  } \text{ subject to} \label{eq:reg-relax-obj} \\
&\sum_{i=1}^n x_i = k, \\
&1 \geq x_i \geq 0 \label{eq:reg-relax-con}
\end{align}
 where \(V(x):=[\sqrt{x_1}v_1 \ldots \sqrt{x_n} v_n]\), to get a fractional solution \(x\in \R^n\). Note that convexity follows from the convexity of function \(\frac{E_{d-1}(M)}{E_d(M)}\) over the set of all PSD matrices \(M\in\R^{n \times n}\). Then, we sample a set \(S\) by \(\lambda\)-regularized proportional volume sampling with hard-core measure \(\mu\), where the parameter \(z\in\R^n\) of the measure \(\mu\) depends on the fractional solution \(x\). The summary of the algorithm is in Algorithm \ref{alg:generalSampleNearInd}. We choose \(z\) in such a way to obtained the desired approximation result. The approximation and motivation to how we set \(z\) can be found in Section \ref{sec:reg-combine-ratio}.

\begin{algorithm}
\caption{Solving \(\min_{S\sbq [n], |S|=k} \frac{E_{d-1} Z_S(\lambda) }{E_{d} Z_S(\lambda) }\) with convex relaxation and \(\lambda\)-regularized proportional volume sampling}\label{alg:generalSampleNearInd}
\begin{algorithmic}[1]
\State Given an input $V=[v_1,\ldots,v_n]$ where $v_i\in \R^d$, $k$ a positive integer, \(\lambda\geq 0.\)
\State Solve the convex relaxation to get a  solution \(x \in \argmin_{x\in[0,1]^n,1^\top x =k}\frac{E_{d-1}\pr{V(x) V(x)^\top + \lambda I}}{E_{d}\pr{V(x) V(x)^\top +  \lambda I}}\). 
\State Let \(z_i = \frac{x_i}{\beta-x_i}\) where \(\beta=1+\frac{\epsilon}{4} \sqrt{1+\frac{\lambda}{\norm{V(x)V(x)}{2}}}\).
\State Sample \(\Sran\) from \(\mu'(S)\propto z^S \det Z_S(\lambda) \) for each \(S \in \U_{\leq k}\).

\State Output $\Sran$ (If $|\Sran| < k$, add $k - |\Sran|$ arbitrary vectors to $\Sran$ first).
\end{algorithmic}
\end{algorithm}

\section{Reduction of Approxibility to Near-Pairwise Independence} \label{sec:reduct}
 In this section, we show that an approximation guarantee of   \(\lambda\)-regularized proportional volume sampling with  measure \(\mu\)  reduces to showing a property on \(\mu\) which we called \textit{near-pairwise independence}, stated formally in Theorem \ref{thm:total-ind-to-approx-with-reg}. We first define \textit{near-pairwise independence} of  a distribution.
\begin{definition}
Let \(\mu\) be a distribution on \(\U \in \{\U_k,\U_{\leq k}\}\). Let \(x\in\R_+^n\). We say \(\mu\) is \textit{(\(c,\alpha\))-near-pairwise independent} with respect to \(x\) if
for all \(T,R\subseteq [n]\) each of size at most \(d\), 
\begin{equation}
\frac{\Prob{\Sran\sim\mu}{S \supseteq T}}{\Prob{\Sran\sim\mu}{S \supseteq R}} \leq c\alpha^{|R|-|T|} \frac{x^T}{x^R} \label{eq:mu-assumption}
\end{equation}
\end{definition}

We omit the phrase "with respect to \(x\)" when the context is clear. Before we prove the main result, we make some calculation which will be used later.
\begin{lemma} \label{lem:expand-det-(d-1)-regularizer}
For any  PSD matrix \(X\in\R^{d\times d}\) and \(a\in\R\),
\begin{align}
E_d\pr{X+aI}=\sum_{i=0}^d E_i(X)a^{d-i}
\end{align}
and
\begin{align}
E_{d-1}\pr{X+aI}=\sum_{i=0}^{d-1}(d-i)E_i(X)a^{d-1-i}
\end{align}
\end{lemma}
\begin{proof}
Let \(\lambda_1,\ldots,\lambda_d\) be eigenvalues of \(X\). Then
we have\begin{align*}
E_d\pr{X+aI}=\prod_{i=1}^d(\lambda_i + a) = \sum_{i=0}^d e_i(\lambda)a^{d-i}= \sum_{i=0}^d E_i(X)a^{d-i}
\end{align*}
which proves the first equality. Next, we have
\begin{align*}
E_{d-1}\pr{X+aI}&=\sum_{j=1}^d\prod_{i\in[d],i\neq j}(\lambda_i + a)  \\
&= \sum_{j=1}^d \sum_{i=0}^{d-1} e_{i}(\lambda_{-j})a^{d-1-i} =\sum_{i=0}^{d-1} \pr{ \sum_{j=1}^d e_{i}(\lambda_{-j})}a^{d-1-i}
\end{align*}
where \(\lambda_{-j}\) is \(\lambda\) with one element \(\lambda_j\) deleted. For each fixed \(i\in\set{0,\ldots,d-1}\), we have 
\begin{equation}
 \sum_{j=1}^d  e_{i}(\lambda_{-j})=(d-i)e_i(\lambda)
\end{equation}
by counting the number of each monomial in \(e_i(\lambda)\). Noting that \(e_i(\lambda)=E_i(X)\), we finish the proof. 
\end{proof} 

Now we are ready to state and prove the main result in this section.
\begin{theorem} \label{thm:total-ind-to-approx-with-reg}
Let \(x\in[0,1]^n\). Let \(\mu\) be a distribution on \(\U \in \{\U_k,\U_{\leq k}\}\) that is (\(c,\alpha\))-near-pairwise independent. Then the \(\lambda\)-regularized proportional volume sampling  \(\mu'\) with measure \(\mu\) satisfies \begin{equation}
\Ex{\Sran\sim\mu'}{\frac{E_{d-1}\pr{Z_\Sran(\lambda)}}{E_{d}\pr{Z_\Sran(\lambda)}}} \leq c\alpha \frac{E_{d-1}\pr{V(x) V(x)^\top + \alpha\lambda I}}{E_{d}\pr{V(x) V(x)^\top + \alpha \lambda I}} \label{eq:alg-ratio-to-convex-relaxation}.
\end{equation}
That is, the sampling gives \(c\alpha\)-approximation guarantee to \(\alpha\lambda\)-regularized \(A\)-optimal design in expectation.
\end{theorem}
Note that by \(\frac{E_{d-1}\pr{V(x) V(x)^\top + \alpha\lambda I}}{E_{d}\pr{V(x) V(x)^\top + \alpha \lambda I}} \leq \frac{E_{d-1}\pr{V(x) V(x)^\top + \lambda I}}{E_{d}\pr{V(x) V(x)^\top + \lambda I}}\), \eqref{eq:alg-ratio-to-convex-relaxation} also implies \(c\alpha\)-approximation guarantee to the original \(\lambda\)-regularized \(A\)-optimal design. However, we can exploit the gap of these two quantities to get a better approximation ratio which converges to 1 as \(\lambda \rightarrow \infty\). This is done formally in Section \ref{sec:reg-combine-ratio}. 

\begin{proof}
We apply Lemma \ref{lem:expand-det-(d-1)-regularizer} to RHS of \eqref{eq:alg-ratio-to-convex-relaxation} to get
\begin{align*}
\frac{E_{d-1}\pr{V(x) V(x)^\top + \alpha\lambda I}}{E_{d}\pr{V(x) V(x)^\top + \alpha \lambda I}} &= \frac{\sum_{h=0}^{d-1}(d-h)E_h(V(x) V(x)^\top)(\alpha \lambda)^{d-1-h}}{ \sum_{\l=0}^dE_\l(V(x) V(x)^\top)(\alpha \lambda)^{d-\l}}\\ 
&=\frac{\sum_{\h=0}^{d-1}\sum_{|T|=\h} (d-\h) (\alpha \lambda )^{d-1-\h} x^T \det\pr{V_T^\top V_T} }{\sum_{\l=0}^d \sum_{|R|=\l} (\alpha \lambda )^{d-\l} x^R \det\pr{V_R^\top V_R}}
\end{align*}
where we apply Cauchy-Binet to the last equality. Next, we apply Lemma \ref{lem:expand-det-(d-1)-regularizer} to LHS of \eqref{eq:alg-ratio-to-convex-relaxation} to get

\begin{align*}
\Ex{\Sran\sim\mu'}{\frac{E_{d-1}\pr{Z_\Sran(\lambda)}}{E_{d}\pr{Z_\Sran(\lambda)}}} &=\frac{\sum_{S\in\U} \mu(S) E_{d} (Z_S(\lambda))\frac{E_{d-1}\pr{Z_\Sran(\lambda)}}{E_{d}\pr{Z_\Sran(\lambda)}} }{\sum_{S\in\U} \mu(S) E_{d}  Z_S(\lambda) } = \frac{\sum_{S\in\U} \mu(S) E_{d-1}  Z_S(\lambda) }{\sum_{S\in\U} \mu(S) E_{d}  Z_S(\lambda) } \\
&= \frac{\sum_{S\in\U} \mu(S)\sum_{h=0}^{d-1}(d-h)E_h(V_S V_S^\top) \lambda^{d-1-h}}{ \sum_{S\in\U} \mu(S)\sum_{\l=0}^dE_\l(V_S V_S^\top)\lambda^{d-\l}} \\
&= \frac{\sum_{S\in\U} \mu(S) \sum_{\h=0}^{d-1}\sum_{|T|=h,T\subseteq S} (d-\h) \lambda^{d-1-\h} \det\pr{V_T^\top V_T}}{\sum_{S\in\U} \mu(S) \sum_{\l=0}^{d}\sum_{|R|=\l,R\subseteq S}  \lambda^{d-\l}  \det\pr{V_R^\top V_R}} \\
&= \frac{\sum_{\h=0}^{d-1} \sum_{|T|=\h} \sum_{S\in\U,S\supseteq T} \mu(S)(d-\h)\lambda^{d-1-\h} \det\pr{V_T^\top V_T}}{\sum_{\l=0}^{d} \sum_{|R|=\l} \sum_{S\in\U,S\supseteq R} \mu(S)\lambda^{d-\l} \det\pr{V_R^\top V_R}} \\
&= \frac{\sum_{\h=0}^{d-1} \sum_{|T|=\h} (d-\h)\lambda^{d-1-\h} \det\pr{V_T^\top V_T}\Prob{\Sran\sim\mu}{\Sran \supseteq T}}{\sum_{\l=0}^{d} \sum_{|R|=\l} \lambda^{d-\l} \det\pr{V_R^\top V_R}\Prob{\Sran\sim\mu}{\Sran \supseteq R}}.
\end{align*}
Therefore, by cross-multiplying the numerator and denominator, the ratio  \(\frac{\Ex{\Sran\sim\mu'}{\frac{E_{d-1}\pr{Z_\Sran(\lambda)}}{E_{d}\pr{Z_\Sran(\lambda)}}} }{\frac{E_{d-1}\pr{V(x) V(x)^\top + \alpha\lambda I}}{E_{d}\pr{V(x) V(x)^\top + \alpha \lambda I}}}\) equals to
\begin{align*}
& \frac{\sum_{\h=0}^{d-1}\sum_{|T|=\h}\sum_{\l=0}^d\sum_{|R|=\l} (d-\h) \det\pr{V_T^\top V_T} \det\pr{V_R V_R^\top} \lambda^{d-1-\h} (\alpha \lambda)^{d-\l} x^R\Prob{\mu}{\Sran \spq T}}{\sum_{\h=0}^{d-1}\sum_{|T|=\h}\sum_{\l=0}^d\sum_{|R|=\l} (d-\h) \det\pr{V_T^\top V_T} \det\pr{V_R V_R^\top} \lambda^{d-\l} (\alpha \lambda)^{d-1-\h} x^T \Prob{\mu}{\Sran \spq R}}.
\end{align*}
For each fixed \(\h,T,\l,R\), we want to upper bound \(\frac{\lambda^{d-1-\h} (\alpha \lambda)^{d-\l} x^R \Prob{\mu}{\Sran \spq T}}{\lambda^{d-\l} (\alpha \lambda)^{d-1-\h} x^T \Prob{\mu}{\Sran \spq R}}\). By the definition of near-pairwise independence \eqref{eq:mu-assumption},
\begin{align}
\frac{\lambda^{d-1-\h} (\alpha \lambda)^{d-\l} x^R \Prob{\mu}{\Sran \spq T}}{\lambda^{d-\l} (\alpha \lambda)^{d-1-\h} x^T \Prob{\mu}{\Sran \spq R}} &\leq  \frac{\lambda^{d-1-\h} (\alpha \lambda)^{d-\l} }{\lambda^{d-\l} (\alpha \lambda)^{d-1-\h} } c\alpha^{\l-\h} \\
&= \alpha^{\h-\l+1} \cdot c\alpha^{\l-\h} = c\alpha
\end{align}
Therefore, the ratio \(\frac{\Ex{\Sran\sim\mu'}{\frac{E_{d-1}\pr{Z_\Sran(\lambda)}}{E_{d}\pr{Z_\Sran(\lambda)}}} }{\frac{E_{d-1}\pr{V(x) V(x)^\top + \alpha\lambda I}}{E_{d}\pr{V(x) V(x)^\top + \alpha \lambda I}}}\) is also bounded above by \(c\alpha\).
\end{proof}

\section{Constructing a  Near-Pairwise-Independent  Distribution} \label{sec:reg-prop-of-mu}
In this section, we want to construct a distribution \(\mu\) on \(\U_{\leq k}\) and prove its (\(c,\alpha\))-near-pairwise-independence.
Our proposed \(\mu\) is hard-core with parameter \(z\in\R^n\) defined by \(z_i := \frac{x_i}{\beta-x_i}\)
(coordinate-wise) for some \(\beta\in(1,2]\) to be chosen later. With this choice of \(\mu\), we upper bound the ratio \(\frac{\Prob{\Sran\sim\mu}{\Sran\spq T}}{\Prob{\Sran\sim\mu}{\Sran\spq R}}\) in terms of \(\beta\). Later in Section \ref{sec:reg-combine-ratio}, after getting an explicit approximation ratio in terms of \(\beta\), we will optimize for \(\beta\)  to get the desired approximation result
to Algorithm \ref{alg:generalSampleNearInd}.
 \cut{The A-optimal paper scales \(x_i\) down to \(\xi_i:=x_i/\beta\) where \(\beta = 1+\epsilon/4\). Then we sample with \(\mu(S) \propto z^S\) where
\begin{equation}
z_i := \frac{\xi_i}{1-\xi_i}
\end{equation}
In this section, we prove the bound with \(\beta\) undefined, and will pick an appropriate \(\beta\) later. }
\begin{lemma} \label{lem:with-reg-mu-approx-ratio}
Let \(x\in [0,1]^n\) such that \(\sum_{i=1}^n x_i = k\). Let \(\mu\) be a distribution on \(\U_{\leq k}\) that is hard-core with parameter \(z\in\R^n\) defined by \(z_i := \frac{x_i}{\beta-x_i}\)
(coordinate-wise) for some \(\beta\in(1,2]\). Then, for all \(T,R\sbq[n]\) of size \(\h,\l\) between 0 and \(d\), we have
\begin{equation}
\frac{\Prob{\Sran\sim\mu}{\Sran\spq T}}{\Prob{\Sran\sim\mu}{\Sran\spq R}} \leq \frac{\beta^{\l-h}}{1-\exp\pr{-\frac{(\beta -1)k-\beta d)^2}{3\beta k}}} \cdot \frac{x^T}{x^R}
.\end{equation}
That is, \(\mu\) is \(\pr{\frac{1}{1-\exp\pr{-\frac{(\beta -1)k-\beta d)^2}{3\beta k}}},\beta}\)-near-pairwise independent.
\end{lemma}
\begin{proof}
Fix \(T,R\) of size \(0\leq \h,\l \leq d\). Define \(\BB\sbq[n]\) to be the random set that includes each \(i\in[n]\) independently with probability \(x_i/\beta\). Let  \(Y_i = \oneIn\br{i\in\BB} \) and  \(Y=\sum_{i\notin R} Y_i\).  Then, noting that \(z_i = \frac{x_i/\beta}{1-x_i/\beta}\), we have
\begin{align*}
\frac{\Prob{\Sran\sim\mu}{\Sran\spq T}}{\Prob{\Sran\sim\mu}{\Sran\spq R}} &=\frac{\Prob{}{\BB\spq T, |\BB| \leq k}}{\Prob{}{\BB\spq R, |\BB| \leq k}}\leq\frac{\Prob{}{\BB\spq T}}{\Prob{}{\BB\spq R, |\BB| \leq k}} \\
&= \beta^{\l-\h}\frac{x^T}{x^R}\frac{1}{\Prob{}{\sum_{i \notin R} Y_i \leq k-\l}}.
\end{align*}
Let \(x(R)=\sum_{i\in R}x_i\). Then by Chernoff bound,
\begin{equation}
\Prob{}{Y> k-\l}\leq \exp\pr{-\frac{\pr{(\beta-1)k+x(R)-\beta \l}^2}{3 \beta (k-x(R))}} \leq \exp\pr{-\frac{\pr{(\beta-1)k-\beta d}^2}{3 \beta k}} 
\end{equation}
which finishes the proof.
\end{proof}

\section{The Proof of the Main Result} \label{sec:reg-combine-ratio}
The main aim of this section is prove the approximation guarantee of the \(\lambda\)-regularized proportional volume sampling algorithm (Algorithm \ref{alg:generalSampleNearInd}) for \(\lambda\)-regularized \(A\)-optimal design. The main result is stated formally in Theorem \ref{thm:reg-ratio-with-lambda}. 

\begin{theorem} \label{thm:reg-ratio-with-lambda}
Let \(V=[v_1,\ldots,v_n]\in\R^{d\times n},\epsilon\in (0,1),\lambda\geq0\), and \(x\in[0,1]^n\), and suppose 
\begin{equation}
k \geq \frac{10 d}{\epsilon} + \frac{60}{\epsilon^2}\log(4/\epsilon). \label{eq:-without-replacement-k-condition}
\end{equation} 
Denote \(\lambda' = \frac{\lambda}{\norm{V(x)V(x)^\top}{2}}\). Then the \(\lambda\)-proportional volume sampling \(\mu'\) with hard-core measure \(\mu\) with parameter \(z_i:=\frac{x_i}{\beta-x_i}\) (coordinate-wise) with \(\beta=1+\frac{\epsilon}{4} \sqrt{1+\lambda'}\) satisfies
\begin{equation}
\Ex{\Sran\sim\mu'}{\frac{E_{d-1}\pr{Z_\Sran(\lambda)}}{E_{d}\pr{Z_\Sran(\lambda)}}} \leq \pr{1+\frac{\epsilon }{\sqrt{1+\lambda'}}} \frac{E_{d-1}\pr{V(x) V(x)^\top + \lambda I}}{E_{d}\pr{V(x) V(x)^\top +  \lambda I}}. \label{eq:bound-frac}
\end{equation}  
Therefore, Algorithm \ref{alg:generalSampleNearInd} gives \((1+\frac{\epsilon }{\sqrt{1+\lambda'}})\)-approximation ratio to   \(\lambda\)-regularized \textit{A}-optimal design.\end{theorem}
The approximation guarantee of Algorithm \ref{alg:generalSampleNearInd} follows from \eqref{eq:bound-frac} because \(x\)  in Algorithm \ref{alg:generalSampleNearInd} is  a convex solution to  \(\lambda\)-regularized \textit{A}-optimal design, so the objective achieved by \(x\) is at most the optimal value of the original problem.

We briefly outline the proof of  Theorem \ref{thm:reg-ratio-with-lambda} here, which combines results from previous sections. Lemma \ref{lem:with-reg-mu-approx-ratio} shows that our constructed \(\mu\) is \((c,\beta)\)-near-pairwise independent for some \(c\) dependent on \(\beta\). Theorem \ref{thm:total-ind-to-approx-with-reg}
converts \((c,\beta)\)-near-pairwise independence to   the \((c\beta\))-approximation guarantee to \(\beta\lambda\)-regularized \(A\)-optimal design.
However, this may be a gap between the optimums of \(\beta\lambda\)- and \(\lambda\)-regularized \(A\)-optimal design. As \(\beta\) increases, the gap is larger so that the approximation tightens even more (we quantify this gap formally in Claim \ref{claim:reg-gap-opt}). As a result, we want to  pick \(\beta\) small enough to have a small \((c\beta\))-approximation ratio but also big enough to exploit this gap. Choosing  \(\beta\) that gives our desired approximation is done in the proof of Theorem \ref{thm:reg-ratio-with-lambda}.

Before proving the main theorem, Theorem \ref{thm:reg-ratio-with-lambda}, we first  simplify the parameter \(c\) of \((c,\beta)\)-near-pairwise independent \(\mu\) that we constructed.
The claim below shows  that \(k=\Omega\pr{\frac{d}{\epsilon}+\frac{\log(1/\epsilon)}{\epsilon^2}}\) is a right condition to obtain \(c\leq 1+\epsilon\).

  \begin{claim} \label{claim:k-condition-general-epsilon}
Let \(\epsilon'>0,\beta>1\). Suppose
\begin{equation}
k \geq \frac{2\beta d}{\beta-1} + \frac{3\beta}{(\beta-1)^2}\log(1/\epsilon').
\label{eq:con-k-d}
\end{equation}
Then
\begin{equation}
\exp\pr{-\frac{(\beta -1)k-\beta d)^2}{3\beta k}} \leq\epsilon' \label{eq:beta-claim}
.\end{equation}
\end{claim}
\begin{proof}
\eqref{eq:beta-claim} is equivalent to
\begin{align*}
(\beta-1)k-\beta d \geq \sqrt{3\beta \log(1/\epsilon') k} 
\end{align*}
which, by solving the quadratic equation in \(\sqrt{k}\), is further equivalent to 
\begin{align*}
\sqrt{k} \geq \frac{ \sqrt{3\beta \log(1/\epsilon')}+\sqrt{3\beta \log(1/\epsilon')+4(\beta-1)\beta d}}{2(\beta-1)}.
\end{align*}
Using inequality \(\sqrt{a}+\sqrt{b}\leq \sqrt{2(a+b)}\), we have
\begin{align*}
\frac{ \sqrt{3\beta \log(1/\epsilon')}+\sqrt{3\beta \log(1/\epsilon')+4(\beta-1)\beta d}}{2(\beta-1)} &\leq \frac{ \sqrt{3\beta \log(1/\epsilon')+2(\beta-1)\beta d}}{\beta-1} \\
&= \sqrt{ \frac{3\beta}{(\beta-1)^2}\log(1/\epsilon')+\frac{2\beta d}{\beta-1} }.
\end{align*}
So, the result follows from \eqref{eq:con-k-d}.
\end{proof}

Next, we quantify the gap of the optimum of \(\beta\lambda\)-regularized \(A\)-optimal design and that of \(\lambda\)-regularized \(A\)-optimal design. 
\begin{claim} \label{claim:reg-gap-opt}
Let \(M\in\R^{d \times d}\) be a PSD matrix, and let \(\beta,\lambda\geq 0\). Then,
\begin{align*}
\frac{E_{d-1}\pr{M + \beta\lambda I}}{E_{d}\pr{M + \beta \lambda I}} \leq\frac{1+\frac{\lambda}{\norm{M}{2}}}{1+ \beta \frac{\lambda}{\norm{M}{2}}} \frac{E_{d-1}\pr{M + \lambda I}}{E_{d}\pr{M + \lambda I}}.
\end{align*}
\end{claim}
\begin{proof}
Let \(\gamma\) be eigenvalues of \(M\). Then,
 \(\frac{\gamma_i+\lambda}{\gamma_i+\beta\lambda} \leq\frac{\norm{M}{2}+\lambda}{\norm{M}{2}+\beta\lambda} =\frac{1+\frac{\lambda}{\norm{M}{2}}}{1+ \beta \frac{\lambda}{\norm{M}{2}}}\) for all \(i\in[d]\). Therefore,
\begin{align*}
\frac{E_{d-1}\pr{M + \beta\lambda I}}{E_{d}\pr{M + \beta \lambda I}}&= \sum_{i=1}^d \frac{1}{\gamma_i+\beta \lambda} \\
&\leq \frac{1+\frac{\lambda}{\norm{M}{2}}}{1+ \beta \frac{\lambda}{\norm{M}{2}}} \sum_{i=1}^d  \frac{1}{\gamma_i+\lambda} =\frac{1+\frac{\lambda}{\norm{M}{2}}}{1+ \beta \frac{\lambda}{\norm{M}{2}}} \frac{E_{d-1}\pr{M + \lambda I}}{E_{d}\pr{M + \lambda I}}
\end{align*}
as desired.
\end{proof}

Now we are ready to prove the main result of this work. 
\begin{proof}[Proof of Theorem \ref{thm:reg-ratio-with-lambda}]
Denote \(\beta_{\lambda'} =1+\frac{\epsilon \sqrt{1+\lambda'}}{4} \) and \(\beta_0=1+\frac{\epsilon}{4}\). By inequality \eqref{eq:-without-replacement-k-condition},  
\begin{align}
k \geq \frac{10 d}{\epsilon} + \frac{60}{\epsilon^2}\log(4/\epsilon)&=  \frac{5d}{2(\beta_0-1)} + \frac{15}{4(\beta_0-1)^2}\log(4/\epsilon)\\ &\geq  \frac{2\beta_0 d}{\beta_0-1} + \frac{3\beta_0}{(\beta_0-1)^2}\log(4/\epsilon). \label{eq:without-replacement-k-con-no-reg}
\end{align}
The last inequality is by \(\beta_0 =1+\frac{\epsilon}{4}\leq \frac{5}{4}\). We have \(\frac{\beta_{0}}{\beta_{0}-1}\geq \frac{\beta_{\lambda'}}{\beta_{\lambda'}-1} \) and \[\frac{\beta_0}{(\beta_0-1)^2} = \frac{1}{\beta_{0}-1}+\frac{1}{(\beta_{0}-1)^2} = \frac{ \sqrt{1+\lambda'}}{\beta_{\lambda'}-1}+\frac{ (\sqrt{1+\lambda'})^2}{(\beta_{\lambda'}-1)^2} \geq \frac{ \sqrt{1+\lambda'}}{\beta_{\lambda'}-1}+\frac{ \sqrt{1+\lambda'}}{(\beta_{\lambda'}-1)^2}\]
\[ =\sqrt{1+\lambda'}\frac{\beta_{\lambda'}}{(\beta_{\lambda'}-1)^2}. \]
Therefore, \eqref{eq:without-replacement-k-con-no-reg} implies
\begin{equation}
k \geq \frac{2\beta_{\lambda'} d}{\beta_{\lambda'}-1} + \frac{3\beta_{\lambda'}}{(\beta_{\lambda'}-1)^2}\sqrt{1+\lambda'}\log(4/\epsilon)
.\end{equation}
By Lemma \ref{lem:with-reg-mu-approx-ratio}, \(\mu\) is \(\pr{c,\beta}\)-near-pairwise independent for \(c=\frac{1}{1-\exp\pr{-\frac{(\beta -1)k-\beta d)^2}{3\beta k}}}\).
We now use Claim \ref{claim:k-condition-general-epsilon} to bound \(c\): with the choice of \(\beta=\beta_{\lambda'} \) and \(\epsilon'=\pr{\frac{\epsilon}{4}}^{\sqrt{1+\lambda'}}\) in Claim \ref{claim:k-condition-general-epsilon}, we have \(c \leq \frac{1}{1-\epsilon'}\).
Therefore, by Theorem \ref{thm:total-ind-to-approx-with-reg}, the objective of Algorithm \ref{alg:generalSampleNearInd}'s output in expectation  is within multiplicative factor \(c\beta=\frac{\beta^{}}{1-\epsilon'}\) from the optimum of \(\beta\lambda\)-regularized \textit{A}-optimal design, i.e.,\begin{equation}
\Ex{\Sran\sim\mu'}{\frac{E_{d-1}\pr{Z_\Sran(\lambda)}}{E_{d}\pr{Z_\Sran(\lambda)}}} \leq \frac{\beta^{}}{1-\epsilon'} \frac{E_{d-1}\pr{V(x) V(x)^\top + \beta\lambda I}}{E_{d}\pr{V(x) V(x)^\top + \beta \lambda I}}. \label{eq:bound-to-beta-lam}
\end{equation}
Now we apply Claim \ref{claim:reg-gap-opt} to exploit the gap between \(\lambda\)- and \(\beta\lambda\)-regularized \textit{A}-optimal designs to get 
\begin{align}
\frac{E_{d-1}\pr{V(x) V(x)^\top + \beta\lambda I}}{E_{d}\pr{V(x) V(x)^\top + \beta \lambda I}} \leq\frac{1+\lambda'}{1+ \beta \lambda'}\cdot \frac{E_{d-1}\pr{V(x) V(x)^\top + \lambda I}}{E_{d}\pr{V(x) V(x)^\top + \lambda I}}. \label{eq:bound-beta-lambda}
\end{align}
Therefore, combining \eqref{eq:bound-to-beta-lam} and \eqref{eq:bound-beta-lambda}, we have that Algorithm \ref{alg:generalSampleNearInd} gives approximation ratio of\begin{align*}
\frac{\beta^{}}{1-\epsilon'} \cdot \frac{1+\lambda'}{1+ \beta \lambda'} &= \pr{1+\frac{\beta-1}{1+\beta \lambda'}} \pr{1-\epsilon'}^{-1} \leq \pr{1+\frac{\beta-1}{1+ \lambda'}} \pr{1-\epsilon'}^{-1} \\
&=\pr{1+\frac{\epsilon }{4\sqrt{1+\lambda'}}}\pr{1-\epsilon'}^{-1}.
\end{align*}
As \(\epsilon/4<1/e\), we have \(\epsilon'=\pr{\frac{\epsilon}{4}}^{\sqrt{1+\lambda'}}\leq \frac{\epsilon}{4\sqrt{1+\lambda'}}\), which gives \( \pr{1-\epsilon'}^{-1} \leq \pr{1-\frac{\epsilon}{4\sqrt{1+\lambda'}}}^{-1} \). Thus, the approximation factor is  bounded by

\begin{equation}
\pr{1+\frac{\epsilon }{4\sqrt{1+\lambda'}}} \pr{1-\frac{\epsilon}{4\sqrt{1+\lambda'}}}^{-1}  \leq 1+\frac{\epsilon }{\sqrt{1+\lambda'}}
\end{equation}
where the inequality is by \(\epsilon \leq 1\).
\end{proof}

\cut{ 
Note that  we could have used fractional solution \(x\) from solving convex relaxation with regularizer \(\beta \lambda\) instead of \(\lambda\) in Algorithm \ref{alg:generalSampleNearInd}. The proof will diverge at \eqref{eq:bound-to-beta-lam}
by applying 
\begin{equation}
\frac{E_{d-1}\pr{V(x) V(x)^\top + \beta\lambda I}}{E_{d}\pr{V(x) V(x)^\top + \beta \lambda I}} \leq\frac{E_{d-1}\pr{V_{S^*}V_{S^*}^\top + \beta\lambda I}}{E_{d}\pr{ V_{S^*}V_{S^*}^\top + \beta \lambda I}}  
\end{equation}
where \(S^*\) is an optimal solution to the original \(\lambda\)-regularized  problem.
The rest of the proof follows similarly, except that we apply Claim  \ref{claim:reg-gap-opt} with \(M=V_{S^*}V_{S^*}^\top\) instead of \(V(x)V(x)^\top\). This gives \((1+\frac{\epsilon }{\sqrt{1+\lambda'}})\)-approximation with \(\lambda'=\frac{\lambda}{\norm{V_{S^*}V_{S^*}^\top}{2}}\) instead of \(\lambda'=\frac{\lambda}{\norm{V(x)V(x)^\top}{2}}\).
\begin{remark}
Let \(V,\epsilon,\lambda,x\) be as in Theorem \ref{thm:reg-ratio-with-lambda}
\end{remark}
}


\section{Efficient Implementation of \(\lambda\)-Regularized Proportional Volume Sampling} \label{sec:efficient-regularized}

In this section, we show that \(\lambda\)-regularized proportional volume sampling can be implemented in polynomial time. In fact, we will show that the same is true for its generalization, \(\lambda\)-regularized proportional \textit{\(l\)-volume }sampling, which is motivated from the generalized ratio objective \eqref{eq:objective-gen}. We first describe proportional \(l\)-volume sampling and its the efficient implementation results. Then, we generalize the results to the \(\lambda\)-regularized counterpart.

An algorithm to solve the generalized ratio objective is proportional {\(l\)-volume }sampling \cite{nikolov2019proportional}, which is to sample \(S\) with probability proportional to \(z^S E_\ell(V_SV_S^\top)\) (instead of \(z^S \det(V_SV_S^\top)\)) for some \(z\in\R^n\) dependent on a fractional solution \(x\in\R^n\) of the convex relaxation of \eqref{eq:objective-gen}. \citet{nikolov2019proportional} show that this algorithm achieves (\(1+\epsilon\))-approximation for \(k\geq\Omega\left(\frac{l}{\epsilon}+\frac{\log 1/\epsilon}{\epsilon^2}\right)\). The efficient implementation of proportional {\(l\)-volume }sampling is stated as follows.
We denote \(O(n^\omega)\) the runtime complexity of matrix multiplication (the best known is \(\omega\approx2.373\) \cite{le2014powers}).
\begin{theorem}[follows from \cite{nikolov2019proportional}] \label{thm:runtime-ori}
Let \(n,d,k\) be positive integers, $z \in\R_{+}^n$, \(\U\in\{ \U_k,\U_{\leq k} \} \), $V=[v_1,\ldots,v_n]\in\R^{d \times n}$, and \(0\leq l' < l \leq d\) be a pair of integers. Let \(\mu'\) be the proportional \(l\)-volume sampling distribution over \(\U\): \(\mu'(S)\propto z^SE_l\pr{V_SV_S^\top}\) for all \(S\in \U\). There are

\begin{itemize}
\item {}an implementation to sample from \(\mu'\)  and
\item a deterministic algorithm that outputs a set $S^*\in \U$  such that
\begin{equation}
\left(
\frac{E_{l'}(V_{S^*} V_{S^*}^\top)}{E_l(V_{S^*} V_{S^*}^\top)} \right)^{\frac{1}{l-l'}}\geq \Expectation{\Sran\sim\mu'}{\left(\frac{E_{l'}(V_\Sran V_\Sran^\top)}{E_l(V_\Sran V_\Sran^\top)}\right)^{\frac{1}{l-l'}}}. \label{eq:deterministic-gen-ineq}
\end{equation}
\end{itemize}
Both algorithms  run in \(O\pr{n^{1+\omega}lk^2\log(lk)}\) number of arithmetic operations.
\end{theorem}

The main ingredient in the algorithms and analysis is to efficiently compute a sum of a particular form efficiently as follows.  
\begin{lemma} 
[follows from \cite{nikolov2019proportional}] \label{sumofProductDet}
Let $z\in\R_+^n,v_1,\ldots,v_n\in \R^d$, and $V=[v_1,\ldots,v_n]$. Let $I,J\subseteq [n]$ be disjoint. Let $1\leq k\leq n$ and \(1\leq l\leq d\). Then the quantity
\begin{equation}\label{exp:hardcoreSum}
 \sum_{|S|=k_0,I\subseteq S,J\cap S=\emptyset}z^S
 E_{d_0}(V_S^\top V_S)
\end{equation}
for all  \(k_0=0,1,\ldots,k\) and \(d_0=0,\ldots,l\) can be simultaneously  computed in \(O\pr{n^\omega lk |I|\cdot\log(lk|I|)}\) number of arithmetic operations. 
\end{lemma}

We outline the proof of Theorem \ref{thm:runtime-ori} briefly here in order to state and prove our result (the full proof of Theorem \ref{thm:runtime-ori} can be found in \cite{nikolov2019proportional}). The proof first shows that, for any given disjoint \(I,J\subseteq[n]\), the marginal probability 
\begin{equation}
P(I,J):=\Prob{\Sran\sim \mu'}{i\in\altmathcal{S}|I\subseteq \Sran, J\cap \Sran=\emptyset} \label{eq:mar-prob}
\end{equation}
and the conditional expectation 
\begin{equation}
X(I,J):=\Expectation{\Sran\sim\mu'}{\left(
\frac{E_{l'}(V_{S^*} V_{S^*}^\top)}{E_l(V_{S^*} V_{S^*}^\top)} \right)^{\frac{1}{l-l'}} | I\subset \altmathcal{S},J\cap \altmathcal{S}=\emptyset} \label{eq:con-ex}
\end{equation}
are in the form of Lemma \ref{sumofProductDet}.   Theorem \ref{thm:runtime-ori} then follows by iteratively sampling an element \(i\in[n]\) one by one with probability \(P(I,J)\) and updating \(I,J\) accordingly. For deterministic algorithm, we compute conditional expectations \(X(I,J)\) for including  and excluding element \(i\), and the smaller choice between the two tells whether to pick \(i\) for that iteration. 

We now state and prove our efficient implementation results.\begin{theorem}
\label{thm:effic-randomized-l-legularized}
Let \(n,d,k\) be positive integers, $z \in\R_{+}^n$, \(\U\in\{ \U_k,\U_{\leq k} \} \), $V=[v_1,\ldots,v_n]\in\R^{d \times n}$, and \(0\leq l' < l \leq d\) be a pair of integers. Let \(\mu'\) be the \(\lambda\)-regularized proportional \(l\)-volume sampling distribution over \(\U\): \(\mu'(S)\propto z^SE_l\pr{V_SV_S^\top+\lambda I_d}\) for all \(S\in \U\). There are

\begin{itemize}
\item {}an implementation to sample from \(\mu'\) and
\item a deterministic algorithm that outputs a set $S^*\in \U$  such that
\begin{equation}
\left(
\frac{E_{l'}(V_{S^*} V_{S^*}^\top+\lambda I_d)}{E_l(V_{S^*} V_{S^*}^\top+\lambda I_d)} \right)^{\frac{1}{l-l'}}\geq \Expectation{\Sran\sim\mu'}{\left(\frac{E_{l'}(V_\Sran V_\Sran^\top+\lambda I_d)}{E_l(V_\Sran V_\Sran^\top+\lambda I_d)}\right)^{\frac{1}{l-l'}}}. \label{eq:deterministic-gen-ineq-lambda}
\end{equation}
\end{itemize}
Both algorithms  run in \(O\pr{n^{1+\omega}lk^2\log(lk)}\) number of arithmetic operations.
\end{theorem}
\begin{proof}
The argument,  similarly to the  proof of Theorem \ref{thm:runtime-ori}, reduces what we need to prove to the ability to efficiently compute marginal probability \eqref{eq:mar-prob} and conditional expectation \eqref{eq:con-ex}.
\
For ease of exposition, we first focus on the marginal probability and \(l=d\). Let \(I'= I \cup \{i\}\). The marginal probability equals to
\begin{align*}
\Prob{\Sran\sim \mu'}{i\in\altmathcal{S}|I\subseteq \Sran, J\cap \Sran=\emptyset}  &= \frac{\Prob{\Sran\sim\mu'}{I'\subseteq \Sran, J\cap \Sran=\emptyset} }{\Prob{\Sran\sim\mu'}{I\subseteq \Sran, J\cap \Sran=\emptyset} }  \\
&= \frac{\sum_{S \in \U,I'\subseteq S,J\cap S=\emptyset}z^{S} \det(V_SV_S^\top+\lambda I_d)}{\sum_{S \in \U,I\subseteq S,J\cap S=\emptyset}z^{S} \det(V_SV_S^\top+\lambda I_d) } \\
&= \frac{\sum_{S \in \U,I'\subseteq S,J\cap S=\emptyset}z^{S}\sum_{h=0}^d \lambda^{d-h}E_h(V_SV_S^\top)}{\sum_{S \in \U,I\subseteq S,J\cap S=\emptyset}z^{S}\sum_{h=0}^d \lambda^{d-h}E_h(V_SV_S^\top)} \\
&= \frac{\sum_{h=0}^d \lambda^{d-h} \sum_{S \in \U,I'\subseteq S,J\cap S=\emptyset}z^{S}E_h(V_SV_S^\top)}{\sum_{h=0}^d \lambda^{d-h}\sum_{S \in \U,I\subseteq S,J\cap S=\emptyset}z^{S}E_h(V_SV_S^\top)}
\end{align*}
where we apply Lemma \ref{lem:expand-det-(d-1)-regularizer} and the Cauchy-Binet formula in the third equality. Both the numerator and denominator are sums over terms in the form \( \sum_{S \in \U, A\subseteq S,J\cap S=\emptyset}z^{S}E_h(V_SV_S^\top)\) for some set \(A\subseteq \U\) and \(h=0,1,\ldots,d\), which by Lemma \ref{sumofProductDet}, can be simultaneously  computed in \(O\pr{n^\omega dk |I|\cdot\log(dk|I|)}\) number of arithmetic operations. (If \(\U=\U_k\), we use Lemma \ref{sumofProductDet} with \(k_0=k\); else if \(\U=\U_{\leq k}\), we use Lemma \ref{sumofProductDet} with all    \(k_0=1,2,\ldots k\).) This runtime is the bottleneck in each of the \(n\) sampling steps, and hence the total runtime is \(O\pr{n^{1+\omega}k^2\log(dk)}\) number of arithmetic operations. 

The key idea in the above argument is in \(\det(V_SV_S^\top+\lambda I)=\sum_{h=0}^d \lambda^{d-h}E_h(V_SV_S^\top)\) in the third equality above, where we expand \(\det(V_SV_S^\top+\lambda I)\) in terms of elementary symmetric polynomials of eigenvalues of \(V_SV_S^\top\). For \(l<d\), the similar argument holds because \(E_l(V_SV_S^\top+\lambda I)\) can still be written as the sum over \(E_h(V_SV_S^\top)\) for several values of \(h\leq l\) (the coefficient may be different from the case \(l=d\), but can be found by a simple counting argument).

We can similarly calculate the conditional expectation and get

\begin{equation}
X(I,J)= \frac{\sum_{S \in \U,I\subseteq S,J\cap S=\emptyset}z^{S} E_{l'}( V_SV_S^\top+\lambda I_d)}{\sum_{S \in \U,I\subseteq S,J\cap S=\emptyset}z^SE_l(V_{S'}V_{S'}^\top+\lambda I_d)}.
\end{equation}
The rest of the proof follows similarly by expanding each elementary symmetric polynomial of eigenvalues of \( V_SV_S^\top+\lambda I_d\) as the sum of elementary symmetric polynomials of eigenvalues of \( V_SV_S^\top\).
  \end{proof}

\section*{Acknowledgement}
The author thanks Mohit Singh (Georiga Institute of Technology) for helpful discussions in the completion of this work.
\bibliographystyle{plainnat}
\bibliography{references}

\appendix
\section{Comparison of Our Bound with \(\lambda\)-Regularized  Volume Sampling}\label{sec:compare}
\citet{derezinski2017subsampling} introduced \textit{\(\lambda\)-regularized  volume sampling}, where we sample a set \(S\subseteq[n]\) of size \(k\) with probability proportional to \(\det(V_SV_S^\top+\lambda I_d)\). They show that for \(\Cov{\eta} \preceq \sigma^2 I\) and \(\lambda\leq \frac{\sigma^2}{\|w^*\|^2}\),
over the expectation of the sampling,\begin{align}
\Ex{\Sran}{\tr\pr{V_{\Sran}V_{\Sran}^\top+ \lambda I}^{-1}} \leq
  \frac{\sigma^2n\,\tr((VV^\top\!\!+\!\lambda I)^{-1})}{k-d_\lambda+1} \label{eq:ridge-old-work-bound}
\end{align}
where \(d_\lambda
  =\tr(V^\top(VV^\top\!\!+\lambda I)^{-1}V)\) \cite{derezinski2017subsampling,derezinski2018reverse}. For \(\lambda=0\), we have \(d_\lambda=d\), and \(d_\lambda\) decreases as \(\lambda\) increases.

The bound \eqref{eq:ridge-old-work-bound} is different from our goal of approximation ratio in this work. Indeed, suppose that  \(S^*\) is an optimal subset of the problem, then in expectation over the run of our Algorithm \ref{alg:generalSampleNearInd},
\begin{align}
\Ex{\Sran}{\tr\pr{V_{\Sran}V_{\Sran}^\top + \lambda I}^{-1}} \leq\pr{1+c\frac{d-1}{(k-d+1)\sqrt{1+\frac{\lambda}{\norm{V(x)V(x)}{2}}}}}
  \sigma^2\tr((V_{S^*}V_{S^*}^\top+\lambda I)^{-1}) \label{eq:ridge-ours-bound}
\end{align}
for some fixed constant \(c\) (we assume  \(d\) is large compared to \(\frac 1\epsilon\) so that \(\frac{d}{\epsilon}+\frac{\log(1/\epsilon)}{\epsilon^2}=O\pr{\frac{d}{\epsilon}}\)). When \(\lambda=0\), our bound \eqref{eq:ridge-ours-bound} simplifies to a bound similar to \eqref{eq:ridge-old-work-bound}:  
\begin{align*}
\Ex{\Sran}{\tr\pr{V_{\Sran}V_{\Sran}^\top + \lambda I}^{-1}} \leq
  \frac{\sigma^2k\,\tr((V_{S^*}V_{S^*}\!\!+\!\lambda I)^{-1})}{k-d_\lambda+1}.
\end{align*}

The main difference between our guarantee and  ones by \cite{derezinski2017subsampling,derezinski2018reverse}  is that ours is in comparison to the best possible subset \(S^*\), whereas \eqref{eq:ridge-old-work-bound} compares the performance to labelling the whole original data set.
Hence, the bound by \cite{derezinski2017subsampling} in worst case may suffer approximation ratio up to an additional factor \(n/k\). 

\section{Calculation of Errors from the Ridge Regression Estimate} \label{sec:cal-error}

\subsection{Calculations in Table   \ref{tab:gau-noise-reg-pred-err}}

The calculations used to obtain distributions in Table   \ref{tab:gau-noise-reg-pred-err} are similar in each of four cases. Here we will compute only one of those, \(X^\top \pr{\hat{w}_S(\lambda) - w^*}\), as steps in its derivation and the result imply the other three. We first state a simple claim that will help in this calculation.
\begin{claim} \label{clm:gau-lin}
For a fixed matrix \(A\) and a random vector \(Z\), we have \(\Cov{AZ}=A\Cov{Z}A^\top\).
\end{claim}
\begin{proof}
Denote \(m=\E{Z}\), the mean vector of \(Z\). Then, the mean of \(AZ\) is \(Am\). We now have
\begin{align*}
\Cov{AZ}&=\E{(AZ-Am)(AZ-Am)^\top} \\
&=\E{A(Z-m)(Z-m)^\top A^\top} \\
&=A\E{(Z-m)(Z-m)^\top}A^\top \\
&=A\Cov{Z}A^\top
\end{align*}
as desired.\end{proof}
We now show how to obtain the distribution of  \(X^\top \pr{\hat{w}_S(\lambda) - w^*}\).
\begin{claim}
We have \[X^\top \pr{\hat{w}_S(\lambda) - w^*}=\NN(-\lambda X^\top  Z_S(\lambda) ^{-1} w^*,\sigma^2 X^\top \br{ Z_S(\lambda) ^{-1} - \lambda Z_S(\lambda) ^{-2} } X) \]
\end{claim}
\begin{proof}
We split the calculation into the following steps.
\begin{enumerate}
\item 
We find the closed-form solution of \(\hat{w}_S(\lambda)\) (e.g. by taking the gradient and set the squared difference to zero) to get \begin{equation}
\hat{w}_S(\lambda) =  Z_S(\lambda) ^{-1} V_S y_S.
\end{equation}
\item
Substituting \(y_i\) from the linear model assumption, we obtain the distribution of the model error as follows.\begin{align*}
\hat{w}_S(\lambda)-w^*&= Z_S(\lambda) ^{-1} V_S y_S-w^* \\
&= Z_S(\lambda) ^{-1} V_S \pr{V_S^\top w^*+\eta_S}-w^* \\
&=  Z_S(\lambda) ^{-1} \br{  Z_S(\lambda)  w^*-(\lambda I)w^*+V_S \eta_S }- w^* \\
&=-\lambda  Z_S(\lambda) ^{-1} w^* +   Z_S(\lambda) ^{-1}  V_S \eta_S.  
\end{align*}
\item
To obtain the distrubtion of the prediction error, we simply left-multiply the above equality by matrix \(X\):
\begin{align}
X^\top \pr{\hat{w}_S(\lambda) - w^*}= -\lambda X^\top  Z_S(\lambda) ^{-1} w^* +X^\top    Z_S(\lambda) ^{-1}  V_S \eta_S. \label{eq:dis-pred}
\end{align}
\item
A linear transformation of a random Gaussian vector is (multi-variate)\ Gaussian, so \eqref{eq:dis-pred} is also Gaussian. We can calculate the mean of  \eqref{eq:dis-pred} as
\begin{equation}
\mu_{X^\top \pr{\hat{w}_S(\lambda) - w^*}} = -\lambda X^\top  Z_S(\lambda) ^{-1} w^* 
\end{equation}
and
the covariance of  \eqref{eq:dis-pred} as\begin{align*}
\Cov{X^\top \pr{\hat{w}_S(\lambda) - w^*}} &=X^\top    Z_S(\lambda) ^{-1}  V_S \Cov{\eta_S} \pr{X^\top    Z_S(\lambda) ^{-1}  V_S }^\top \\
&=X^\top    Z_S(\lambda) ^{-1}  V_S \Cov{\eta_S} V_S^\top   Z_S(\lambda) ^{-1} X \\
&= \sigma^2 X^\top    Z_S(\lambda) ^{-1}  V_S V_S^\top   Z_S(\lambda) ^{-1} X    \\
 &=\sigma^2 X^\top \br{ Z_S(\lambda) ^{-1} - \lambda Z_S(\lambda) ^{-2} } X
\end{align*}
where we use  Claim \ref{clm:gau-lin} for the first equality. We note that  if \(\Cov{\eta}\preceq \sigma^2 I_n\) instead of \(\Cov{\eta}= \sigma^2 I_n\), the third equality is replaced by "\(\preceq\)", so the errors we need to bound in this work is no more than those when \(\Cov{\eta}= \sigma^2 I_n\). 
\end{enumerate}
\end{proof}
\subsection{Calculations in Table \ref{tab:sqr-loss-reg-pred-err}}
We use the notation \((x)_i\) to denote the \(i\)th coordinate of the vector \(x\). First, we calculate expected squared distance of the model (or predictor)\ error:
 \begin{align*}
\Ex{\eta}{\norm{\hat{w}_S(\lambda) - w^*}{2}^2} &= \sum_{i=1}^d \Ex{\eta}{\pr{(\hat{w}_S(\lambda))_i - (w^*)_i}^2} \\
&= \sum_{i=1}^d \pr{ \Ex{\eta}{(\hat{w}_S(\lambda))_i - (w^*)_i}^2 + \Var{(\hat{w}_S(\lambda))_i - (w^*)_i} }\\
&=\norm{\Ex{\eta}{{\hat{w}_S(\lambda) - w^*}}}{2}^{2} + \tr{\Cov{\hat{w}_S(\lambda) - w^*}}
\end{align*}
where we use \(\Ex{}{X^2}=\Ex{}{X}^2+\Var{X} \) (bias-variance decomposition). Similarly, for prediction error, \begin{align*}
\Ex{\eta}{\norm{X^\top\pr{\hat{w}_S(\lambda) - w^*}}{2}^2} &= \sum_{i=1}^d \Ex{\eta}{\norm{(X^\top\pr{\hat{w}_S(\lambda))_i - (w^*)_i}}{2}^2} \\
&= \sum_{i=1}^d \pr{ \Ex{\eta}{X^\top\pr{(\hat{w}_S(\lambda))_i - (w^*)_i}}^2 + \Var{X^\top \pr{(\hat{w}_S(\lambda))_i - (w^*)_i} }} \\
&=\norm{\Ex{\eta}{{X^\top\pr{\hat{w}_S(\lambda) - w^*}}}}{2}^2+ \tr{\Cov{X^\top\pr{\hat{w}_S(\lambda) - w^*}}.}
\end{align*}
As we know the mean and variance of the distributions of model and prediction errors (summarized in Table \ref{tab:gau-noise-reg-pred-err}), we can substitute those means and variances to obtain
\begin{align*}
\Ex{\eta}{\norm{\hat{w}_S(\lambda) - w^*}{2}^2}&=\norm{-\lambda Z_S(\lambda) ^{-1} w^*}{2}^2+ \tr{ \sigma^2 \br{ Z_S(\lambda) ^{-1} - \lambda Z_S(\lambda) ^{-2} }} \\
&= \lambda^2 \an{ Z_S(\lambda) ^{-2},w^* {w^*}^\top}+\sigma^2 \tr{ Z_S(\lambda) ^{-1}}-\lambda \sigma^2\tr{ Z_S(\lambda) ^{-2}} \\
&= \sigma^2 \tr{ Z_S(\lambda) ^{-1}}-\lambda\an{ Z_S(\lambda) ^{-2}, \sigma^2I-\lambda w^* {w^*}^\top}
\end{align*}
 and\begin{align*}
\Ex{\eta}{\norm{X^\top\pr{\hat{w}_S(\lambda) - w^*}}{2}^2} &=\norm{-\lambda X^\top  Z_S(\lambda) ^{-1} w^*}{2}^2+ \tr{ \sigma^2 X^\top \br{ Z_S(\lambda) ^{-1} - \lambda Z_S(\lambda) ^{-2} } X} \\
&= \lambda^2 \an{ Z_S(\lambda) ^{-1}X X^\top Z_S(\lambda)^{-1}   ,w^* {w^*}^\top}+\sigma^2 \tr{ X^\top Z_S(\lambda) ^{-1}X}-\lambda\sigma^2 \tr{X^\top Z_S(\lambda) ^{-2} X} \\
&=\sigma^2 \tr{ X^\top Z_S(\lambda) ^{-1}X} -\lambda\an{ Z_S(\lambda) ^{-1}X X^\top Z_S(\lambda)^{-1}   ,\sigma^2I-\lambda w^* {w^*}^\top}.
\end{align*}

Note that, similar to \eqref{eq:model-error-final-bound}, if we assume that 
\( \lambda \leq \frac{\sigma^2}{\norm{w^*}{2}^2}\), then we have\begin{align}
\Ex{\eta}{\norm{X^\top\pr{\hat{w}_S(\lambda) - w^*}}{2}^2} \leq\sigma^2 \tr{ X^\top Z_S(\lambda) ^{-1}X,}  \label{eq:prediction-error-final-bound}
\end{align}
the prediction-error version of the model-(or predictor-)error bound   \eqref{eq:model-error-final-bound}.
\end{document}

%% file: header.tex
\usepackage[utf8]{inputenc}
\usepackage[T1]{fontenc}
\usepackage[]{algorithm}
\usepackage[noend]{algpseudocode}
\usepackage{amsmath, amssymb, amsthm}
\usepackage{bm}
\usepackage{graphicx}
\usepackage{hyperref}
\usepackage{thmtools}
\usepackage{fullpage}
\usepackage{mathtools}

\usepackage{bbm}

\DeclareMathAlphabet{\altmathcal}{OMS}{cmsy}{m}{n}

\newcommand{\norm}[2]{\ensuremath{\left\lVert {#1} \right\rVert_{#2}}}

\newcommand{\Prob}[2]{\underset{#1}{\ensuremath{\mathsf{Pr}}}\left[#2\right]}
\newcommand{\Expectation}[2]{\underset{#1}{\mathbb{E}} \left[ #2  \right]}

\newcommand{\E}[1]{\ensuremath{\mathbb{E}}\text{{$\left[#1\right]$}}}

\def\R{\mathbb{R}}
\def\BB{\altmathcal{B}}
\def\NN{\mathcal{N}}

\def\U{\altmathcal{U}}

\newcommand{\Sran}{\altmathcal{S}}

\newcommand{\cut}[1]{}

\newcommand{\Cov}[1]{\operatorname{Cov}\pr{#1}}
\newcommand{\Var}[1]{\operatorname{Var}\pr{#1}}

\newcommand{\bc}[1]{\left\{#1\right\}}
\usepackage{tabls}
\usepackage{tabularx, booktabs}
\newcolumntype{x}[1]{>{\centering\arraybackslash\hspace{0pt}}p{#1}}
\newcommand{\Ex}[2]{\underset{#1}{\mathbb{E}} \left[ #2  \right]}
\newcommand{\oneIn}{\mathbbm{1}}
\renewcommand{\l}{\lcurve}
\newcommand{\lcurve}{\ell}
\newcommand{\sbq}{\subseteq}
\newcommand{\spq}{\supseteq}
\def\h{h}

\usepackage{xspace}
\usepackage{wrapfig,boxedminipage}
\usepackage{framed}
\usepackage{subcaption}

\newtheorem{claim}{Claim}

\def\vcomment#1{{\color{red}}}

\def\R{\mathbb{R}}
\newcommand{\eps}{\epsilon}
\newcommand{\br}[1]{\left[#1\right]}
\newcommand{\pr}[1]{\left(#1\right)}
\newcommand{\set}[1]{\ensuremath{\left\{#1\right\}}}
\newcommand{\an}[1]{\left\langle#1\right\rangle}
\newcommand{\tr}{\operatorname{tr}}

\newcommand{\argmin}{\operatorname*{argmin}}

\def\rank{\hbox{\rm rank}}

\theoremstyle{plain}
\newtheorem{theorem}{Theorem}[section]
\newtheorem{lemma}[theorem]{Lemma}

\theoremstyle{definition}
\newtheorem{definition}[theorem]{Definition}

\newtheorem{remark}[theorem]{Remark}